\def\year{2022}\relax
   \documentclass[letterpaper]{article}
   \usepackage{aaai22} 
   \usepackage{times} 
   \usepackage{helvet} 
   \usepackage{courier} 
   \usepackage[hyphens]{url} 
   \usepackage{graphicx} 
   \usepackage{graphicx}  
   \usepackage{natbib}  
   \usepackage{caption}  
   \DeclareCaptionStyle{ruled}%
     {labelfont=normalfont,labelsep=colon,strut=off}
\title{Fixation Maximization in the Positional Moran Process}

\author{
Joachim Brendborg,\textsuperscript{\rm 1}
Panagiotis Karras,\textsuperscript{\rm 1}
Andreas Pavlogiannis,\textsuperscript{\rm 1}\\
Asger Ullersted Rasmussen,\textsuperscript{\rm 1}
Josef Tkadlec\textsuperscript{\rm 2}
}
\affiliations{
\textsuperscript{\rm 1}Aarhus University, Aabogade 34, Aarhus, Denmark\\
\textsuperscript{\rm 2}Harvard University, 1 Oxford Street, Cambridge, USA\\
\{panos, pavlogiannis\}@cs.au.dk, tkadlec@math.harvard.edu
}

\def\mypath{}

\usepackage{xfrac} 
\usepackage{amsmath, amsthm,amssymb,amsfonts}
\usepackage{thmtools}
\usepackage{thm-restate}

\usepackage[capitalise]{cleveref}

\usepackage{tikz}
\usetikzlibrary{positioning,arrows,automata,shapes,decorations,decorations.markings,calc, matrix,decorations.pathmorphing, patterns,backgrounds}

\def\O{\mathcal{O}}
\def\bigO{\mathcal{O}}

\def\Pr{\mathbb{P}}
\def\Nats{\mathbb{N}}
\def\R{\mathbb{R}}

\def\d{\operatorname{d}}
\def\dfp{\fp'(G^S,0)}

\DeclareMathOperator*{\argmax}{arg\,max}

\newtheorem{theorem}{Theorem}

\newtheorem{corollary}{Corollary}

\newcommand{\TotalFitness}{\mathsf{F}}
\newcommand{\Fitness}{\mathsf{f}}
\newcommand{\X}{\mathsf{X}}

\newcommand{\fp}{\operatorname{fp}}
\newcommand{\ft}{\operatorname{T}}
\newcommand{\E}{\mathbb{E}}

\newcommand{\Temp}{\mathcal{T}}
\newcommand{\Traj}{\mathcal{T}}
\newcommand{\Deltafp}{\operatorname{\Delta fp}}

\newcommand{\NodeActivationMoran}{\texttt{FM}}
\newcommand{\NodeActivationMoranStrong}{\texttt{FM}^{\infty}}
\newcommand{\NodeActivationMoranWeak}{\texttt{FM}^{0}}

\newcommand{\Opt}{OPT}

\newcommand{\Paragraph}[1]{{\smallskip\noindent\bf #1}}
\newcommand{\SubParagraph}[1]{{\smallskip\noindent\em #1}}
\newcommand{\FitAdv}{\delta}

\newcommand{\NP}{\operatorname{NP}}

\newcommand{\NPH}{\operatorname{NP-hard}}

\newcommand{\PTime}{\operatorname{PTime}}
\newcommand{\Weight}{w}

\newcommand\numberthis{\addtocounter{equation}{1}\tag{\theequation}}

\newcommand{\MC}{\mathcal{M}}

\begin{document}
\maketitle

\begin{abstract}
The Moran process is a classic stochastic process that models spread dynamics on graphs.
A single ``mutant'' (e.g., a new opinion, strain, social trait etc.) invades a population of ``residents'' scattered over the nodes of a graph.
The mutant fitness advantage $\FitAdv\geq 0$ determines how aggressively mutants propagate to their neighbors.
The quantity of interest is the fixation probability, i.e., the probability that the initial mutant eventually takes over the whole population.
However, in realistic settings, the invading mutant has an advantage only in certain locations.
E.g., the ability to metabolize a certain sugar is an advantageous trait to bacteria only when the sugar is actually present in their surroundings.
In this paper, we introduce the \emph{positional Moran process}, a natural generalization in which the mutant fitness advantage is only realized on specific nodes called \emph{active} nodes, and study the problem of \emph{fixation maximization}:~given a budget~$k$, choose a set of~$k$ active nodes that maximize the fixation probability of the invading mutant.
We show that the problem is $\NPH$, while the optimization function is not submodular, indicating strong computational hardness. 
We then focus on two natural limits:
at the limit of $\FitAdv\to\infty$ (\emph{strong selection}), although the problem remains $\NPH$, the optimization function becomes submodular and thus admits a constant-factor greedy approximation; at the limit of~$\FitAdv\to 0$ (\emph{weak selection}), we show that we can obtain a tight approximation in 
$O(n^{2\omega})$ time, where
$\omega$ is the matrix-multiplication exponent. An experimental evaluation of the new algorithms along with some proposed heuristics corroborates our results.
\end{abstract}

\section{Introduction}

Several real-world phenomena are described in terms of an \emph{invasion process} that occurs in a network and follows some \emph{propagation model}. 
The problems considered in such settings aim to optimize the success of the invasion, quantified by some metric. 
For example, social media campaigns serving purposes of marketing, political advocacy, and social welfare are modeled in terms of \emph{influence spread}, whereby the goal is to maximize the expected influence spread of a message as a function of a set of initial adopters~\cite{Kempe2003, Youze:2014, Borgs:2014:MSI, Tang:2015:IMN, Zhang2020}; see~\cite{IM-Survey} for a survey. 
Similarly, in rumor propagation~\cite{Demers1987}, the goal is to minimize the number of rounds until a rumor
spreads to the whole population~\cite{Fountoulakis2012}.

A related propagation model is that of the Moran process, introduced as a model of genetic evolution~\cite{Moran1958}, and the variants thereof, such as the discrete Voter model~\cite{Clifford1973,Liggett1985, Antal2006, Talamali2021}. 
A single \emph{mutant} (e.g., a new opinion, strain, social trait, etc.) invades a population of \emph{residents} scattered across the nodes of the graph. 
Contrariwise to influence and rumor propagation models, the Moran process accounts for \emph{active} resistance to the invasion:~a node carrying the mutant trait may not only forward it to its neighbors, but also lose it by receiving the resident trait from its neighbors. 
This model exemplifies settings in which individuals may switch opinions several times or sub-species competing in a gene pool.
A key parameter in this process is the \emph{mutant fitness advantage}, a real number~$\FitAdv\geq 0$ that expresses the intensity by which mutants propagate to their neighbors compared to residents. Large~$\FitAdv$ favors mutants, while~$\FitAdv=0$ renders mutants and residents indistinguishable.

Success in the Moran process is quantified in terms of the \emph{fixation probability}, i.e., the probability that a randomly occurring mutation takes over the whole population. 
As the focal model in evolutionary graph theory~\cite{Lieberman2005}, the process has spawned extensive work on identifying \emph{amplifiers}, i.e., graph structures that enhance the fixation probability~\cite{Monk2014,Giakkoupis16,Galanis2017,Pavlogiannis2018,Tkadlec2021}. 
Still, in the real world, the probability that a mutation spreads to its neighborhood is often affected by its position:
the ability to metabolise a certain sugar is advantageous to bacteria only when that sugar is present in their surroundings;
likewise, people are likely to spread a viewpoint more enthusiastically
if it is supported by experience from their local environment (country, city, neighborhood). 
By default, the Moran process does not account for such positional effects.
While optimization questions have been studied in the Moran and related models~\cite{EvenDar2007}, the natural setting of positional advantage remains unexplored. 
In this setting, the question arises: to which positions (i.e., nodes) should we confer an advantage so as to maximize the fixation probability?

\Paragraph{Our contributions.}
We define a \emph{positional} variant of the Moran process, in which the mutant fitness advantage is only realized on a subset of nodes called \emph{active} nodes and an associated optimization problem, \emph{fixation maximization} ($\NodeActivationMoran$): given a graph~$G$, a mutant fitness advantage~$\FitAdv$ and a budget~$k$, choose a set of~$k$ nodes to activate that maximize the fixation probability~$\fp(G^S,\FitAdv)$ of mutants. We also study the problems~$\NodeActivationMoranStrong$ and~$\NodeActivationMoranWeak$ that correspond to~$\NodeActivationMoran$ at the natural limits~$\FitAdv\to\infty$ and~$\FitAdv\to 0$, respectively. 
\emph{On the negative side}, we show that $\NodeActivationMoran$ and $\NodeActivationMoranStrong$ are~$\NPH$ (\cref{subsec:np_hardness}). One common way to circumvent $\NP$-hardness is to prove that the optimization function is submodular; however, we show that $\fp(G^S,\FitAdv)$ is not submodular when~$\FitAdv$ is finite (\cref{subsec:non_submodularity}). 
\emph{On the positive side}, we first show that the fixation probability $\fp(G^S,\FitAdv)$ on unidrected graphs admits a FPRAS, that is, it can be approximated to arbitrary precision in polynomial time (\cref{subsec:approx}). 
We then show that, in contrast to finite~$\FitAdv$, the function $\fp^\infty(G^S)$ is submodular on undirected graphs and thus  $\NodeActivationMoranStrong$ admits a polynomial-time constant-factor approximation (\cref{subsec:positive_strong}). 
Lastly, regarding the limit~$\FitAdv\to0$, we show that $\NodeActivationMoranWeak$ can be solved in polynomial time on any graph (\cref{subsec:positive_weak}). 
Overall, while $\NodeActivationMoran$ is hard in general, we obtain tractability in both natural limits $\FitAdv\to\infty$ and $\FitAdv\to0$.

Due to space constraints, some proofs are in the~\cref{sec:appendix}. 

\section{Preliminaries}
We extend the standard Moran process to its positional variant and define our fixation maximization problem on it.
\subsection{The Positional Moran Process}\label{subsec:moran_process}

\Paragraph{Structured populations.} Standard evolutionary graph theory~\cite{Nowak2006} describes a population structure by a directed weighted graph (network)~$G=(V,E, \Weight)$ of~$|V|=n$ nodes, where~$\Weight$ is a weight function over the edge set~$E$ that defines a probability distribution on each node~$u$, i.e., for each~$u\in V$, $\sum_{(u,v)\in E}\Weight(u,v)=1$.
We require that~$G$ is strongly connected when~$E$ is projected to the support of the probability distribution of each node. Nodes represent sites (locations) and each site is occupied by a single agent. Agents are of two types --- \emph{mutants} and \emph{residents} --- and agents of the same type are indistinguishable from each other.
For ease of presentation, we may refer to a node and its occupant agent interchangeably. 
In special cases, we will consider simple undirected graphs, meaning that (i)~$E$ is symmetric and (ii)~$\Weight(u,\cdot)$ is uniform for every node~$u$.

\SubParagraph{Active node sets and fitness.} Mutants and residents are differentiated in \emph{fitness}, which in turn depends on a set~$S \!\subseteq\! V$ of \emph{active nodes}. 
Consider that mutants occupy the nodes in~$\X \!\subseteq\! V$. 
The \emph{fitness} of the agent at node~$u \!\in\! V$ is:
\begin{align}
\Fitness_{\X,S}(u)=
\begin{cases}
1+\FitAdv, & \text{if } u\in \X\cap S\\
1,         & \text{if } u\not \in \X\cap S
\end{cases}
\end{align}
where $\FitAdv\geq 0$ is the \emph{mutant fitness advantage}. 
Intuitively, fitness expresses the capacity of a type to spread in the network.
The resident fitness is normalized to~$1$, while~$\FitAdv$ measures the relative competitive advantage conferred by the mutation. 
However, the mutant fitness advantage is not realized uniformly in the graph, but only at active nodes. 
The total fitness of the population is $\TotalFitness_{S}(\X) = \sum_{u}\Fitness_{\X,S}(u)$.

\Paragraph{The positional Moran process.} We introduce the \emph{positional Moran process} as a discrete-time stochastic process $(\X_t)_{t\geq 0}$, where each $\X_t\subseteq V$ is a random variable denoting the set of nodes of $G$ occupied by mutants. Initially, the whole population consists of residents; at time~$t=0$, a single mutant appears uniformly at random. That is, $\X_0$ is given by $\Pr[\X_0=\{u\}]=1/n$ for each $u\in V$. In each subsequent step, we obtain $\X_{t+1}$ from $\X_{t}$ by executing the following two stochastic events in succession.
\begin{enumerate}
\item \emph{Birth event:} A single agent is selected for reproduction with probability proportional to its fitness, i.e., an agent on node~$u$ is chosen with probability $\Fitness_{\X_t, S}(u)/\TotalFitness_{S}(\X_{t})$.
\item \emph{Death event:} A neighbor~$v$ of~$u$ is chosen with probability~$\Weight(u,v)$; its agent is replaced by a copy of the one at~$u$.
\end{enumerate}
If~$u$ is a mutant and~$v$ a resident, then mutants spread to~$v$, hence~$\X_{t+1} = \X_t \cup \{v\}$; likewise, if~$u$ is a resident and~$v$ a mutant, then~$\X_{t+1} = \X_t \setminus \{v\}$; otherwise $u$ and~$v$ are of the same type, and thus~$\X_{t+1}=\X_{t}$.

\SubParagraph{Related Moran processes.} 
Our positional Moran process generalizes the standard Moran process on graphs~\cite{Nowak2006}
that forms the special case of~$S \!=\! V$ (that is, the mutant fitness advantage holds on the whole graph).
Conversely, it can be viewed a special case of the Moran process with two environments~\cite{Kaveh2019}.
It is different from models in which agent's fitness depends on which agents occupy the neighboring nodes
(so-called frequency dependent fitness)~\cite{Huang2010}.

\Paragraph{Fixation probability.} 
If $\exists\tau \ge 0$ such that $X_{\tau}=V$, we say that the mutants \emph{fixate} in~$G$, otherwise we say they \emph{go extinct}. 
As~$G$ is strongly connected, for any initial mutant set~$A\subseteq V$, the process reaches a homogeneous state almost surely, i.e.,~$\lim_{t\to\infty} \Pr[\X_{t}\in \{\emptyset, V\}] = 1$. 
The \emph{fixation probability (FP) of such a set~$A$ is}~$\fp(G^S,\FitAdv, A) = \Pr[\exists\tau \ge 0\colon \X_{\tau} \!=\! V \mid \X_{0} \!=\! A]$. 
Thus, the FP of the mutants in~$G$ with active node set~$S$ and fitness advantage~$\FitAdv$ is:
\begin{align}
\fp(G^S,\FitAdv)=\frac{1}{n}\sum_{u\in V}\fp(G^S,\FitAdv, \{u\})
\end{align}
For any fixed~$G$ and~$S$, the function~$\fp(G^S,\FitAdv)$ is continuous in~$\FitAdv$ and infinitely differentiable: the underlying Markov chain is finite and its transition probabilities are rational functions in~$\FitAdv$, so~$\fp(G^S,\FitAdv)$ too is a rational function in $\FitAdv$, and since the process is defined for any~$\FitAdv\in[0,\infty)$, $\fp(G^S, \FitAdv)$ has no points of discontinuity for~$\delta\in[0,\infty)$.

\SubParagraph{Computing the fixation probability.} 
In the neutral setting where $\FitAdv \!=\! 0$, $\fp(G^S,\FitAdv) = \frac{1}{n}$ for any graph~$G$~\cite{Broom2010}.
When~$\FitAdv>0$, the complexity of computing the fixation probability is unknown even when~$S \!=\! V$; it is open whether the function can be computed in~$\PTime$. For undirected~$G$ and~$S \!=\! V$, the problem has a fully polynomial randomized approximation scheme (FPRAS)~\cite{Diaz2014, Chatterjee2017, Goldberg2020}, as the expected absorption time of the process is polynomial and thus admits fast Monte-Carlo simulations. 
We also rely on Monte-Carlo simulations for computing~$\fp(G^S,\FitAdv)$, and show that the problem admits a FPRAS for any~$S \!\subseteq\! V$ when~$G$ is undirected.

\begin{figure}[!t]
\includegraphics[width=0.95\linewidth]{\mypath 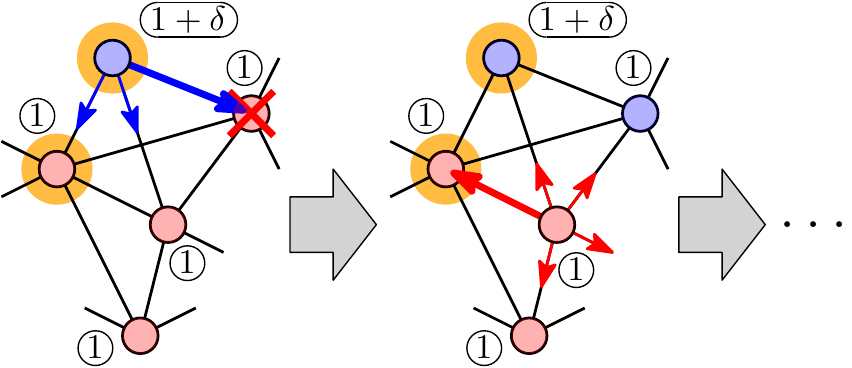}
\caption{Steps in the positional Moran process. Blue (red) nodes are occupied by mutants (residents). Yellow rings indicate active nodes that realize the mutant fitness advantage.}
\label{fig:moran}
\end{figure}

\subsection{Fixation Maximization via Node Activation}\label{subsec:node_activation}

We study the optimization problem of \emph{fixation maximization} in the positional Moran process. We focus on three variants based on different regimes of the mutant fitness advantage~$\FitAdv$.

\Paragraph{Fixation maximization.} Consider a graph~$G \!=\! (V,E)$, a mutant fitness advantage~$\FitAdv$ and a budget~$k \!\in\! \Nats$. The problem \emph{Fixation Maximization for the Positional Moran Process} $\NodeActivationMoran(G,\FitAdv,k)$ calls to determine an active-node set of size~$\leq k$ that maximizes the fixation probability:
\[
S^*=\argmax_{S\subseteq V, |S|\leq k} \fp(G^S,\FitAdv)
\numberthis\label{eq:optimization}
\]
As we will argue, $\fp(G^S,\FitAdv)$ is monotonically increasing in~$S$, hence the condition~$|S| \!\leq\! k$ can be replaced by~$|S| \!=\! k$. The associated decision problem is: Given a budget~$k$ and a threshold~$p\in [0,1]$, determine whether there exists an active-node set~$S$ of size~$|S|=k$ such that~$\fp(G^S,\FitAdv)\geq p$. As an illustration, consider a cycle graph on $n=50$ nodes and two strategies: one that activates~$k$ ``spaced'' nodes, and another that activates~$k$ ``contiguous'' nodes. As \cref{fig:cycle} shows, depending on~$\delta$ and~$k$, either strategy could be better.

\Paragraph{The regime of strong selection.} Next, we consider $\NodeActivationMoran$ in the limit $\FitAdv\to\infty$. The corresponding fixation probability is: 
\begin{align}
\fp^{\infty}(G^S)=\lim_{\FitAdv\to\infty}\fp(G^S,\FitAdv)
\end{align}
As we show later, $\fp(G^S,\FitAdv)$ is monotonically increasing in~$\FitAdv$, and since it is also bounded by~1, the above limit exists. The corresponding optimization problem $\NodeActivationMoranStrong(G,k)$ asks for the active node set~$S^*$ that maximizes $\fp^{\infty}(G^S)$:
\[
S^*=\argmax_{S\subseteq V, |S|\leq k} \fp^{\infty}(G^S)
\numberthis\label{eq:optimization_strong}
\]

\begin{figure}[!t]
\includegraphics[width=0.95\linewidth]{\mypath 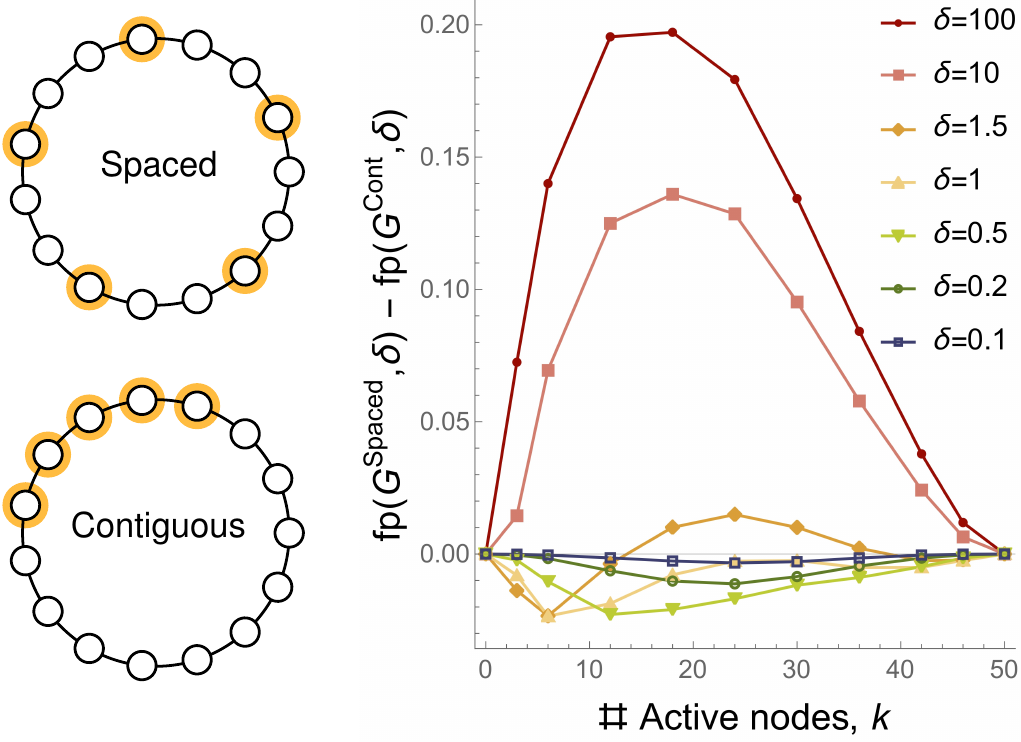}
\caption{FP differences on a cycle graph on $n=50$ nodes, under spaced and contiguous activating strategy. E.g. for $\delta=100$ and $k=18$ we have $\fp(G^{\text{Spaced}},\delta)\doteq 0.62$ and $\fp(G^{\text{Cont}},\delta)\doteq 0.43$, so the difference is almost $20\%$.}
\label{fig:cycle}
\end{figure}

\Paragraph{The regime of weak selection.} Further, we consider the problem $\NodeActivationMoran(G,\FitAdv,k)$ in the limit $\FitAdv\to0$. For small positive~$\FitAdv$, different sets~$S$ generally give different fixation probabilities~$\fp(G^S,\FitAdv)$, see~\cref{fig:tristar}. For fixed~$S \!\subseteq\! V$, we can view~$\fp(G^S,\FitAdv)$ as a function of~$\FitAdv$. 
By the Taylor expansion of~$\fp(G^S,\FitAdv)$, and given that~$\fp(G^S,0)=\sfrac1n$, we have
\begin{align}
\fp(G^S,\FitAdv) = \frac{1}{n} + \FitAdv\cdot \dfp  +\O(\FitAdv^2),
\end{align}
where $\dfp \!=\! \frac{\d}{\d\FitAdv}\Bigr|_{\substack{\FitAdv=0}}\fp(G^S,\FitAdv)$. 
Since lower-order terms~$\O(\FitAdv^2)$ tend to~$0$ faster than~$\FitAdv$ as~$\FitAdv \!\to\! 0$, on a sufficiently small neighborhood of~$\FitAdv \!=\! 0$, maximizing~$\fp(G^S,\FitAdv)$ reduces to maximizing the derivative~$\dfp$ (up to lower order terms). The corresponding optimization problem~$\NodeActivationMoranWeak(G,k)$ asks for the set~$S^*$ that maximizes~$\dfp$:
\[
S^*=\argmax_{S\subseteq V, |S|\leq k} \dfp
\numberthis\label{eq:optimization_weak}
\]
As $\fp(G^S,\FitAdv)\geq \sfrac1n$ for all~$S$, our choice of~$S$ maximizes the \emph{gain} of fixation probability, i.e., the difference~$\Deltafp(G^S,\FitAdv) = \fp(G^S,\FitAdv) - \frac{1}{n}$. The set~$S^*$ of~\cref{eq:optimization_weak} guarantees a gain with relative approximation factor tending to~$1$ as~$\FitAdv \to 0$, i.e., we have the following (simple) lemma.

\begin{restatable}{lemma}{lem:weakselapproxguarantee}\label{lem:weak_sel_approx_guarantee}
$\frac{\Deltafp(G^{S^{\Opt}},\FitAdv)}{\Deltafp(G^{S^*},\FitAdv)}=1+\O(\FitAdv)$.
\end{restatable}

\begin{figure}[t]
\includegraphics[width=0.95\linewidth]{\mypath 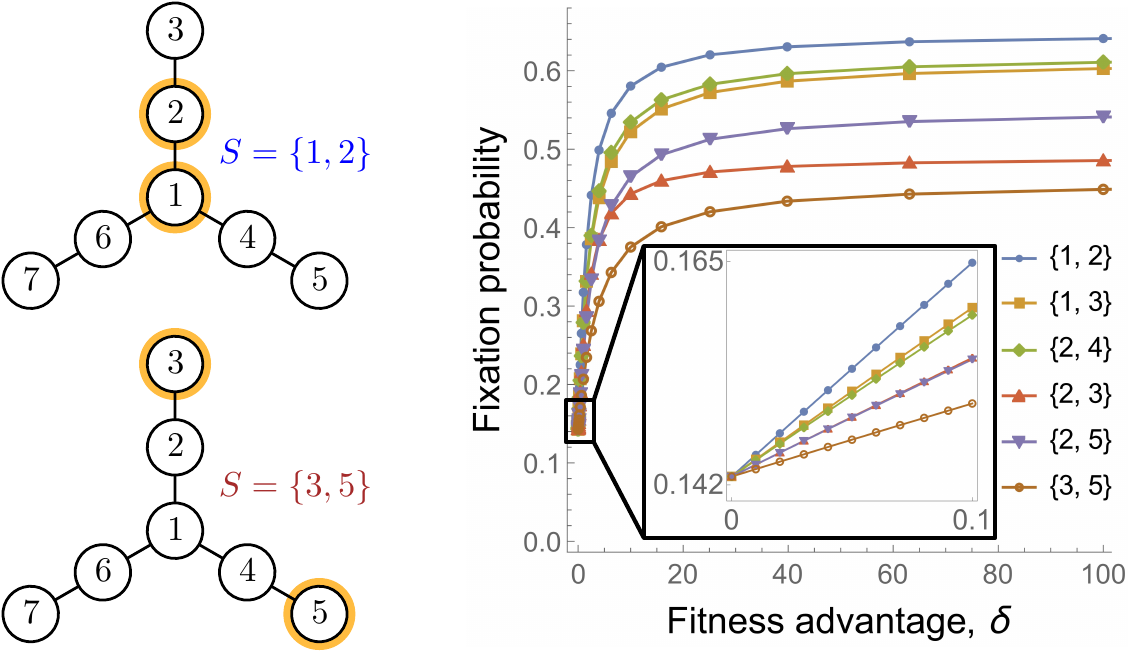}
\caption{A fixed graph~$G$ on 7 nodes and fixation probability~$\fp(G^S,\FitAdv)$ for all six non-equivalent subsets~$S$ of~$k=2$ nodes; $\fp(G^S,\FitAdv)$ is roughly linear in~$\FitAdv$ when~$\FitAdv\in[0,0.1]$.}
\label{fig:tristar}
\end{figure}

\Paragraph{Monotonicity.} 
Enlarging the set~$S$ or the advantage~$\FitAdv$ increases the fixation probability. The proof generalizes the standard Moran process monotonicity~\cite{Diaz2016}.

\begin{restatable}[Monotonicity]{lemma}{lemmonotonicity}\label{lem:monotonicity}
Consider a graph $G=(V,E)$, two subsets $S\subseteq S'\subseteq V$ and two real numbers $0\le \FitAdv\le \FitAdv'$. Then
$ \fp(G^S,\FitAdv)\le \fp(G^{S'},\FitAdv')$.
\end{restatable}

\section{Negative Results}\label{sec:hardness}
We start by presenting hardness results for $\NodeActivationMoran$ and $\NodeActivationMoranStrong$.

\subsection{Hardness of $\NodeActivationMoran$ and $\NodeActivationMoranStrong$}\label{subsec:np_hardness}

We first show $\NP$-hardness for $\NodeActivationMoran$ and $\NodeActivationMoranStrong$.
This hardness persists even given an oracle that computes the fixation probability $\fp(G^S,\FitAdv)$ (or $\fp^{\infty}(G^S)$) for a graph~$G$ and a set of active nodes~$S$.
Given a set of mutants $A$, let
\begin{align}
\fp^{\infty}(G^S,A)=\lim_{\FitAdv\to\infty}\fp(G^S,\FitAdv, A)
\end{align}
be the fixation probability of a set of mutants $A$ under strong selection, i.e., in the limit of large fitness advantage. The following lemma states that for undirected graphs, a single active mutant with infinite fitness guarantees fixation.

\begin{restatable}{lemma}{lemfixprobstrong}\label{lem:fixprobstrong}
Let $G$ be an undirected graph, $S$ a set of active nodes, and $A$ a set of mutants. If  $A\cap S\neq \emptyset$ then $\fp^{\infty}(G^S, A)=1$.
\end{restatable}

The intuition behind \cref{lem:fixprobstrong} is as follows:~consider a node $u\in A\cap S$.
Then, as~$\FitAdv\to\infty$, node~$u$ is chosen for reproduction at a much higher rate than any resident node.
Hence, in a typical evolutionary trajectory, the neighbors of node $u$ are occupied by mutants for most of the time and thus the individual at node~$u$ is effectively guarded from any threat of becoming a resident.
In contrast, the mutants are reasonably likely to spread to any other part of the graph.
In combination, this allows us to argue that mutants fixate with high probability as~$\FitAdv\to\infty$.
Notably, this argument relies on~$G$ being undirected.
As a simple counterexample, consider~$G$ being a directed triangle~$a\to b\to c\to a$ and~$S = A = \{a\}$. 
Then, for the mutants to spread to~$c$, a mutant on~$b$ must reproduce. 
Since~$b\not \in S$, a mutant on~$b$ is as likely to reproduce in each round as a resident on~$c$. 
If the latter happens prior to the former, node~$a$ becomes resident, and, with no active mutants left, we have a constant (non-zero) probability of extinction.

Due to \cref{lem:fixprobstrong}, if~$G$ is an undirected regular graph with a vertex cover of size~$k$, under strong selection, we achieve the maximum fixation probability by activating a set~$S$ of~$k$ nodes that forms a vertex cover of~$G$. The key insight is that, if the initial mutant lands on some node~$u\not \in S$, then the mutant fixates if and  only if it manages to reproduce once before it is replaced, as all neighbors of~$u$ are in~$S$ and thus \cref{lem:fixprobstrong} applies. This is captured in the following lemma.

\begin{restatable}{lemma}{lemvertexcover}\label{lem:vertex_cover}
For any undirected regular graph $G$ and $S\subseteq V$, $\fp^{\infty}(G^{S})\geq\frac{n+|S|}{2n}$ iff $S$ is a vertex cover of $G$ (and if so, the equality holds).
\end{restatable}

\begin{proof} 
First, note that due to \cref{lem:fixprobstrong}, the fixation probability equals the probability that eventually a mutant is placed on an active node. Due to the uniform placement of the initial mutant, the probability that it lands on an active node is $\frac{|S|}{n}$. 
Let $A$ be the set of nodes in~$V\setminus S$ that have at least one neighbor not in~$S$.
Thus we can write:
\begin{align}
\fp^{\infty}(G^{S}) &= {\textstyle \frac{|S|}{n} \!+\! \frac{n-|S|-|A|}{n}p} \!+\! \sum_{u\in A}{\textstyle \frac{1}{n}} \fp^{\infty}(G^S\!, \{u\})
\label{eq:vc_upper_bound}
\end{align}
where~$p$ is the probability that an initial mutant occupying a non-active node~$u$ whose neighbors are all in~$S$ spreads to any of its neighbors, i.e., reproduces before any of its neighbors replace it; $u$ reproduces with probability~$p_1 = 1/n$, while its neighbors replace it with probability $p_2 = \sum_{(v,u)\in E} 1/n\cdot 1/d = d \cdot 1/n\cdot 1/d = 1/n$, where~$d$ is the degree of~$G$, hence~$p = p_1/(p_1+p_2) = 1/2$. 
If~$S$ is a vertex cover of~$G$, then $A = \emptyset$, hence \cref{eq:vc_upper_bound} yields~$\fp^{\infty}(G^{S}) = \frac{n+|S|}{2n}$.
If~$S$ is not a vertex cover of~$G$, then~$|A|\geq 1$; for any node~$u \in A$, since at least one of its neighbors is not in~$S$, $\fp^{\infty}(G^S, \{u\})$ is strictly smaller than the probability that the mutant on~$u$ reproduces before it gets replaced, which, as we argued, is~$\frac{1}{2}$, hence~\cref{eq:vc_upper_bound} yields~$\fp^{\infty}(G^{S})<\frac{n+|S|}{2n}$. 
\end{proof}

The hardness for $\NodeActivationMoranStrong$  is a direct consequence of \cref{lem:vertex_cover} and the fact that vertex cover is $\NPH$ even on regular graphs (see, e.g.,~\cite{Feige2003}). Moreover, as $\fp(G^S, \FitAdv)$ is a continuous function on~$\FitAdv$, we also obtain hardness for $\NodeActivationMoran$ (i.e., under finite $\FitAdv$).
We thus have the following theorem.

\begin{restatable}{theorem}{thmnphard}\label{thm:np_hard}
$\NodeActivationMoranStrong(G,k)$ and $\NodeActivationMoran(G,\FitAdv,k)$ are $\NPH$, even on undirected regular graphs.
\end{restatable}

Finally, we note that the hardness persists even given an oracle that computes the fixation probability $\fp(G^S,\FitAdv)$ (or $\fp^{\infty}(G^S)$) given a graph~$G$ and a set of active nodes~$S$.


\subsection{Non-Submodularity for $\NodeActivationMoran$}\label{subsec:non_submodularity}

One standard way to circumvent $\NP$-hardness is to show that the optimization function is submodular, which implies that a greedy algorithm offers constant-factor approximation guarantees~\cite{Nemhauser1978,Krause2014}. Although in the next section we prove that $\fp^\infty(G^S)$ is indeed submodular, unfortunately the same does not hold for $\fp(G^S, \FitAdv)$, even for very simple graphs.

\begin{restatable}{theorem}{thmnonsubmodular}\label{thm:non_submodular}
The function $\fp(G^S,\FitAdv)$ is not submodular.
\end{restatable}
\begin{proof}
Our proof is by means of a counter-example.
Consider $G=K_{4}$, that is, $G$ is a clique on~$4$ nodes.
For $i\in\{0,1,2\}$ denote by $p_i$ the fixation probability on $K_4$ when $i$ of the 4 nodes are active and mutants have fitness advantage $\FitAdv=1/3$.
For fixed $i$, the value $p_i$ can be computed exactly
by solving a system of $2^4$ linear equations (one for each configuration of mutants and residents).
Solving the three systems using a computer algebra system, we obtain $p_0=1/4$, $p_1=38413/137740<0.2788$ and $p_2=28984/94153>0.3078$.
Hence $p_0+p_2>0.5578=2\cdot 0.2789>2p_1$, thus submodularity is violated.
\end{proof}

For $G=K_4$, the three systems can even be solved symbolically, in terms of $\FitAdv>0$. It turns out that the submodularity property is violated for $\FitAdv<\delta^\star$, where $\delta^\star\doteq 0.432$. In contrast, for $G=K_3$ the submodularity property holds for all $\FitAdv\ge 0$ (the verification is straightforward, though tedious).

\section{Positive Results}\label{sec:positive}
We now turn our attention to positive results for fixation maximization in the limits of strong and weak selection.

\subsection{Approximating the Fixation Probability}\label{subsec:approx}

We first focus on computing the fixation probability.
Although the complexity of the problem is open even in the standard Moran process,
when the underlying graph is undirected, the expected number of steps until
the standard (non-positional) process terminates 
is polynomial in $n$~\cite{Diaz2014}, which yields a FPRAS.
A straightforward extension of that approach applies to the expected number $\ft(G^S,\FitAdv)$ of steps in the positional Moran process.
\begin{restatable}{lemma}{thmapproxfixprob}\label{thm:approx_fix_prob}
Given an undirected graph $G=(V,E)$ on $n$ nodes,
a set $S\subseteq V$ and a real number $\FitAdv\geq 0$,
we have $\ft(G^S,\FitAdv)\le (1+\FitAdv)n^6$.
\end{restatable}

As a consequence, by simulating the process multiple times and reporting the proportion of runs that terminated with mutant fixation, we obtain a fully polynomial randomized approximation scheme (FPRAS) for the fixation probability.

\begin{corollary}\label{cor:approx_fix_prob}
Given a connected undirected graph $G=(V,E)$, a set $S\subseteq V$ and a real number $\FitAdv\geq 0$,
the function $\fp(G^S,\FitAdv)$ admits a FPRAS.
\end{corollary}

\subsection{Fixation Maximization under Strong Selection}\label{subsec:positive_strong}

Here we show that the function $\fp^\infty(G^S)$ is submodular.
As a consequence, we obtain a constant-factor approximation for $\fp^\infty(G^S)$ in polynomial time.

\begin{restatable}[Submodularity]{lemma}{lemstrongselectionsubmodular}
\label{lem:strong_selection_submodular}
For any undirected graph $G$, the function $\fp^\infty(G^S)$ is submodular. 
\end{restatable}
\begin{proof}
Consider any set $S$ of active nodes.
Then the positional Moran process $\{\X_t\}$ with mutant advantage $\FitAdv\to\infty$ is equivalent to the following process:
\begin{enumerate}
\item While $\X_t\cap S=\emptyset$, perform an update step $\X_t\to\X_{t+1}$ as if $\FitAdv=0$.
\item If $\X_t\cap S\ne \emptyset$, terminate and report fixation.
\end{enumerate}
Indeed, as long as no active node hosts a mutant, all individuals have the same (unit) fitness and
the process is indistinguishable from the Moran process with $\FitAdv=0$.
On the other hand, once any one active node receives a mutant, fixation happens with high probability by~\cref{lem:fixprobstrong}.
All in all, the fixation probability $p(S)=\fp^{\infty}(G^S)$ in the limit $\FitAdv\to\infty$
can be computed by simulating the neutral process ($\FitAdv=0$) until
either $\X_t=\emptyset$ (extinction) or $\X_t\cap S\ne\emptyset$ (fixation).

To prove the submodularity of $p(S)$, it suffices to show that for any two sets $S,T\subseteq V$ we have:
\begin{align}
p(S)+p(T)\ge p(S\cup T) + p(S\cap T).
\end{align}
Consider any fixed trajectory $\Traj = (X_0,X_1,\dots)$.
We say that a subset $U\subseteq V$ is \emph{good} with respect to $\Traj$, or simply \emph{good}, if there exists a time point~$\tau\ge 0$ such that $\X_{\tau}\cap U\ne\emptyset$. It suffices to show that, for any $\Traj$, there are at least as many good sets among $S$, $T$ as they are among $S\cup T$ and $S\cap T$. To that end, we distinguish three cases based on how many of the sets $S\cup T$, $S\cap T$ are good.
\begin{enumerate}
\item Both of them: Since $S\cap T$ is good, both $S$ and $T$ are good and we get $2\ge 2$.
\item One of them: If $S\cap T$ is good we conclude as before (we get $2\ge 1$).
Otherwise $S\cup T$ is good, hence at least one of $S$, $T$ is good and we get $1\ge 1$.
\item None of them: We have $0\ge 0$. \qedhere
\end{enumerate}
\end{proof}

The submodularity lemma leads to the following approximation guarantee for the greedy algorithm~\cite{Nemhauser1978,Krause2014}.

\begin{restatable}{theorem}{thmstorngselection}\label{thm:storng_selection}
Given an undirected graph~$G$ and integer~$k$, let~$S^*$ be the solution to $\NodeActivationMoranStrong(G,k)$, and~$S'$ the set returned by a greedy maximization algorithm.
Then~$\fp^{\infty}(G^{S'})\geq \left(1-\frac{1}{e}\right)\fp^{\infty}(G^{S^*})$.
\end{restatable}

Consequently, a greedy algorithm approximates the optimal fixation probability within a factor of~$1-1/e$, provided it is equipped with an oracle that estimates the fixation probability $\fp^{\infty}(G^{S})$ given any set~$S$. 
For the class of undirected graphs, \cref{thm:approx_fix_prob} and \cref{thm:storng_selection} yield a fully polynomial randomized greedy algorithm for $\NodeActivationMoranStrong(G,k)$ with approximation guarantee $1-1/e$.

\subsection{Fixation Maximization under Weak Selection}\label{subsec:positive_weak}

Here we show that $\NodeActivationMoranWeak$ can be solved in polynomial time.
Recall that given an active set $S\subset V$, we can write $\fp(G^S,\FitAdv)=\fp(G^S,0)+\FitAdv\cdot \dfp + \bigO(\FitAdv^2)$, while $\NodeActivationMoranWeak(G,k)$ calls to compute a set~$S^\star$ which maximizes the coefficient $\dfp$ across all sets~$S$ that satisfy~$|S|\le k$. We build on the machinery developed in~\cite{Allen2017}.

In~\cref{lem:weak_selection} below, we show that $\fp'(G^S,0)=\sum_{u_i\in S} \alpha(u_i)$
for a certain function $\alpha\colon V\to\R$.
The bottleneck in computing $\alpha(\cdot)$ is
solving a system of
$\bigO(n^2)$
linear equations,
one for each pair of nodes.
This can be done in matrix-multiplication time 
$\bigO(n^{2\omega})$
for $\omega<2.373$~\cite{Alman2021}.
By virtue of the linearity of~$\fp'(G^S,0)$, the optimal set~$S^\star$ is then given by
choosing the~top-$k$ nodes~$u_i$ in terms of~$\alpha(u_i)$.
Thus we establish the following theorem.

\begin{theorem}\label{thm:weak_selection}
Given a graph $G$ and some integer~$k$, the solution to $\NodeActivationMoranWeak(G,k)$ can be computed in
$\bigO(n^{2\omega})$
time, where $\omega<2.373$ is the matrix multiplication exponent.
\end{theorem}

In the rest of this section, we state~\cref{lem:weak_selection}
and give intuition about the function $\alpha(\cdot)$.
For the formal proof of~\cref{lem:weak_selection}, see~\cref{sec:appendix}. 

For brevity, we denote the $n$ nodes of $G$ by $u_1,\dots,u_n$
and we denote by~$p_{ij}=\Weight(u_i, u_j)$ the probability that,
if~$u_i$ is selected for reproduction, then the offspring migrates to~$u_j$.

\begin{restatable}{lemma}{thmweakselection}\label{lem:weak_selection}
Let $G=(V,E)$ be a graph with $n$ nodes $u_1,\dots,u_n$. 
 Consider a function $\alpha\colon V\to\R$ defined by $\alpha(u_i) =  \frac{1}{n}\cdot \sum_{j\in[n]} p_{ij} \cdot \pi_j\cdot \psi_{ij}$,
where $\{\pi_i\mid u_i\in V\}$ is the solution to the linear system
given by $\sum_{i\in [n]} \pi_i=1$ together with
\begin{align}
\pi_i= \left(1-\sum_{j\in[n]} p_{ji}\right)\cdot \pi_i +\sum_{j\in[n]} p_{ij}\cdot \pi_j\quad\forall i\in[n],
\label{eq:pi2}
\end{align}
and $\{\psi_{ij} \mid (u_i,u_j)\in E\}$ is the solution to the linear system
\begin{align}
\psi_{ij}=
\begin{cases}
\frac{1+\sum_{\ell\in[n]}\left(p_{\ell i}\cdot \psi_{\ell j} + p_{\ell j}\cdot \psi_{i\ell}\right)}{\sum_{\ell\in[n]}\left(p_{\ell i}+p_{\ell j}\right)} & i\neq j 
\\
0, & \text{otherwise}.
\end{cases}
\label{eq:psi}
\end{align}
Then $\fp'(G^S,0)=\sum_{u_i\in S} \alpha(u_i)$.
\end{restatable}

The intuition behind the quantities $\pi_i$, $\psi_{ij}$, and $\alpha(u_i)$ from~\cref{lem:weak_selection} is as follows:
Consider the Moran process on $G$ with $\FitAdv=0$.
Then for~$i\in[n]$, the value~$\pi_i$ is in fact the mutant fixation probability
starting from the initial configuration $\X_0=\{u_i\}$~\cite{Allen2021}.
Indeed, at any time step, the agent on~$u_i$ will eventually fixate if
(i)~$u_i$ is not replaced by its neighbors and eventually fixates from the next step onward, or
(ii)~$u_i$ spreads to a neighbor~$u_j$ and fixates from there.
The first event happens with rate $(1-\sum_{j\in[n]}p_{ji})\cdot\pi_i $, 
while the second event happens with rate $\sum_{j\in[n]} p_{ij}\cdot \pi_j$.
(We note that for undirected graphs, the system has an explicit solution
$\pi_i = \frac{\sfrac{1}{\deg(u_i)}}{\sum_{j\in[n]}\sfrac{1}{\deg(u_j)}}$~\cite{Broom2010}.)

The values $\psi_{ij}$ for
$i\ne j$
also have intuitive meaning:
They are equal to the expected total number of steps during the Moran process (with $\FitAdv=0$)
in which $u_i$ is mutant and $u_j$ is not, 
when starting from a random initial configuration $\Pr[\X_0=\{u_i\}]=1/n$.
To sketch the idea behind this claim,
suppose that in a single step, a node $u_\ell$ spreads to node $u_i$
(this happens with rate $p_{\ell i}$).
Then the event ``$u_i$ is mutant while $u_j$ is not'' holds
if and only if $u_\ell$ was mutant but $u_j$ was not, hence the product $p_{\ell i}\cdot \psi_{\ell j}$.
Similar reasoning applies to the product $p_{\ell j}\cdot \psi_{i\ell}$.
The term~$1$ in the numerator comes out of the $1/n$ probability that the initial mutant lands on $u_i$ (in which case indeed $u_i$ is mutant and $u_j$ is not).

Finally, the function $\alpha(\cdot)$ also has a natural interpretation:
The contribution of an active node~$u_i$ to the fixation probability grows vs.~$\FitAdv$ at $\FitAdv = 0$
to the extent that a mutant~$u_i$ is likely to be chosen for reproduction,
spread to a resident neighbor~$u_j$, and fixate from there;
that is, by the total time that~$u_i$ is mutant but a neighbor~$u_j$ is not~($\psi_{ij}$),
calibrated by the probability that, when chosen for reproduction,~$u_i$ propagates to that neighbor~($p_{ij}$),
who thereafter achieves mutant fixation~($\pi_j$);
this growth rate is summed over all neighbors and weighted by the rate~$1/n$ at which~$u_i$ reproduces when~$\FitAdv=0$.

\section{Experiments}\label{sec:experiments}

Here we report on an experimental evaluation of the proposed algorithms and some other heuristics. 
Our data set consists of~110 connected subgraphs of community graphs and social networks of the Stanford Network Analysis Project~\cite{Leskovec2014}.
These subgraphs where chosen randomly, and varied in size between 20-170 nodes.
Our evaluation is not aimed to be exhaustive, but rather to outline the practical performance of various heuristics, sometimes in relation to their theoretical guarantees.

In particular, we use 7 heuristics, each taking as input a graph, a budget $k$, and (optionally) the mutant fitness advantage $\FitAdv$.
The heuristics are motivated by our theoretical results and by the related fields of influence maximization and evolutionary graph theory.
\begin{figure*}[!ht]
\centerline{
\includegraphics[scale=0.75]{\mypath 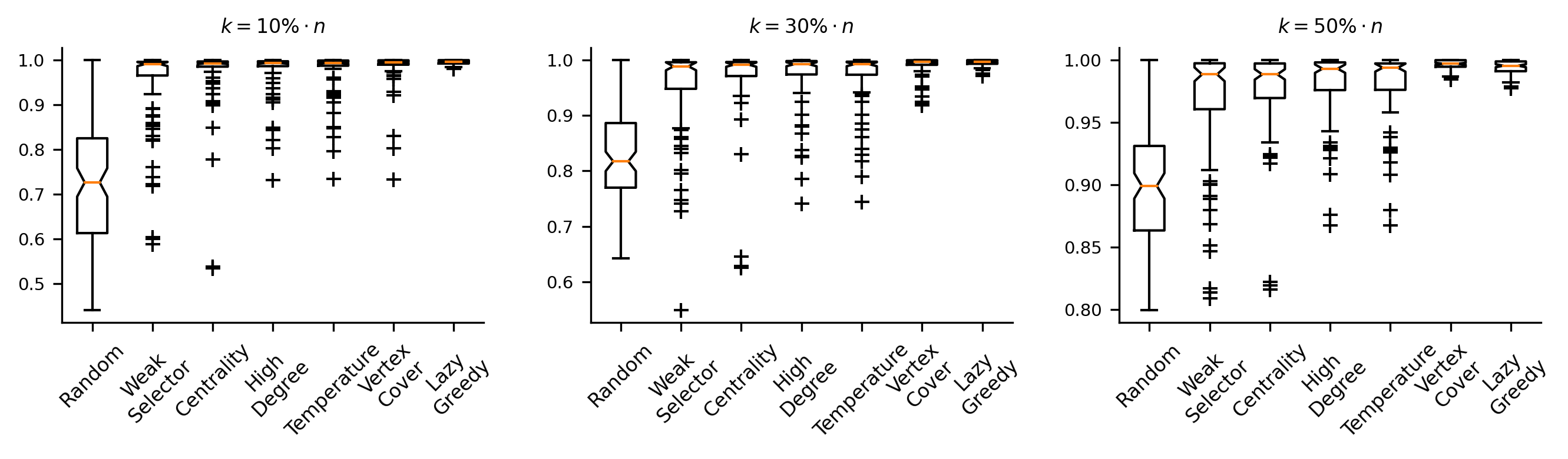}
}
\caption{Heuristic performance for $\NodeActivationMoranStrong$.}\label{fig:experiments_strong}
\end{figure*}

\begin{figure*}[!ht]
\includegraphics[scale=0.75]{\mypath 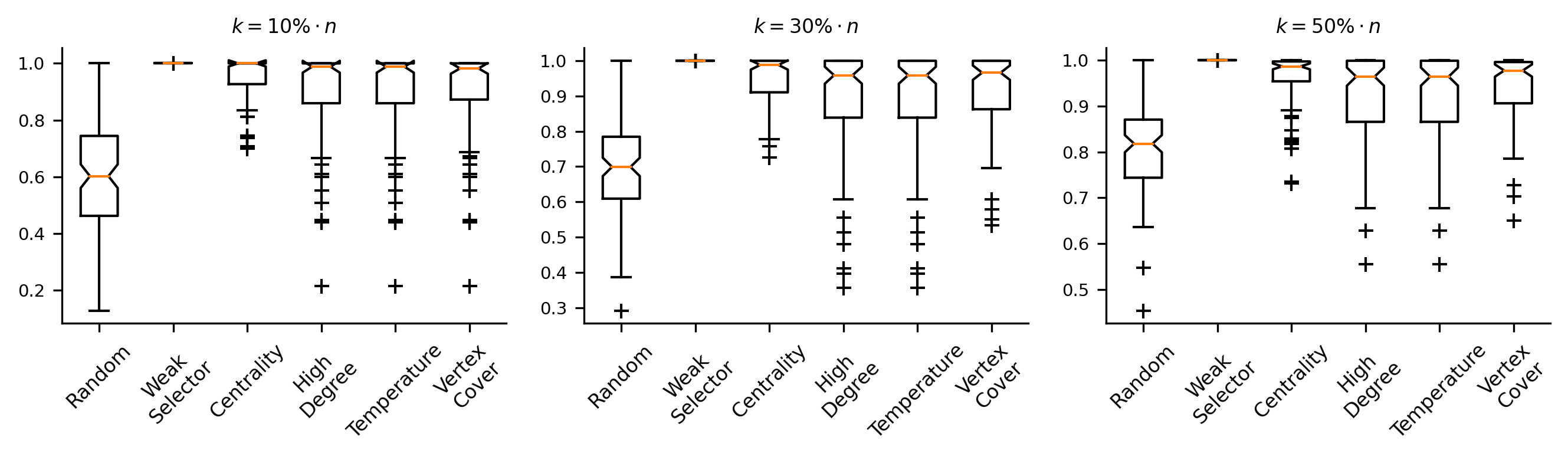}
\caption{Heuristic performance for $\NodeActivationMoranWeak$.}\label{fig:experiments_weak}
\end{figure*}

\begin{enumerate}
\item \emph{Random:}~Simply choose~$k$ nodes uniformly at random. This heuristic serves as a baseline.
\item \emph{High Degree:}~Choose the~$k$ nodes with the largest degree.
\item \emph{Centrality:}~Choose the~$k$ nodes with the largest betweenness centrality.
\item \emph{Temperature:}~The temperature of a node~$u$ is defined as $\Temp(u)=\sum_{v} \Weight(v,u)$; this heuristic chooses the~$k$ nodes with the largest temperature.
\item \emph{Vertex Cover:}~Motivated by \cref{lem:vertex_cover}, this heuristic attempts to maximize the number of edges with at least one endpoint in $S$.
In particular, given a set $A\subseteq V$, let 
\begin{align}
c(A)=\{ |(u,v)\in E\colon u\in A\text{ or }v\in A|  \}.
\end{align}
The heuristic greedily maximizes $c(\cdot)$, i.e.,
we start with $S=\emptyset$ and perform $k$ update steps
\begin{align}
S\gets S\cup  \argmax_{u\not \in S} c(S\cup\{u\}).
\end{align}
\item \emph{Weak Selector:}~Choose the~$k$ nodes that maximize $\dfp$ using the (optimal) weak-selection method.
\item \emph{Lazy Greedy:}~A simple greedy algorithm starts with~$S=\emptyset$, and in each step chooses the node~$u$ to add to~$S$ that maximizes the objective function.
For $\NodeActivationMoran(G,\FitAdv,k)$ and $\NodeActivationMoranStrong(G,k)$ this process requires to repeatedly evaluate $\fp(G^{S\cup\{u\}},\FitAdv)$ (or $\fp(G^{S\cup\{u\}},\FitAdv)$) for every node $u$.
This is done by simulating the process a large number of times (recall \cref{thm:approx_fix_prob}), and becomes a computational bottleneck when we require high precision. As a workaround we suggest a lazy variant of the greedy algorithm~\cite{Minoux1978}, which is faster but requires submodularity of the objective function. In effect, this algorithm is a correct implementation of the greedy heuristic in the limit of strong selection (recall \cref{lem:strong_selection_submodular}), while it may still perform well (but without guarantees) for finite~$\FitAdv$.
\end{enumerate}

In all cases ties are broken arbitrarily.
Note that the heuristics vary in the amount of information they have about the graph and the invasion process.
In particular, Random has no information whatsoever, High Degree only considers the direct neighbors of a node,
Temperature considers direct and distance-2 neighbors of a node, Centrality and Vertex Cover consider the whole graph, while Weak Selector and Lazy Greedy are the only ones informed about the Moran process.

For each graph $G$, we have chosen values of~$k$ corresponding to $10\%$, $30\%$ and $50\%$ of its nodes, and have evaluated the above heuristics in their ability to solve $\NodeActivationMoranStrong(G,K)$ (strong selection) and $\NodeActivationMoranWeak(G,k)$ (weak selection). 
We have not considered other values of~$\FitAdv$ as evaluating $\fp(G^S,\FitAdv)$ precisely via simulations requires many repetitions and becomes slow.

\Paragraph{Strong selection.}
We start with the case of $\NodeActivationMoranStrong(G,K)$.
Since different graphs $G$ and budgets $k$ generally result in highly variant fixation probabilities, in order to get an informative aggregate measure of each heuristic, we divide the fixation probability it obtains by the maximum fixation probability obtained across all heuristics for the same graph and budget. This normalization yields values in the interval $[0,1]$, and makes comparison straightforward. \cref{fig:experiments_strong} shows our results. 
We see that the Lazy Greedy algorithm performs best for small budgets ($10\%$, $30\%$), while its performance is matched by Vertex Cover for budget $50\%$.
The high performance of Lazy Greedy is expected given its theoretical guarantee (\cref{thm:storng_selection}).
We also observe that, apart from Weak Selector, the other heuristics perform quite well on many graphs, though there are several cases on which they fail, appearing as outlier points in the box plots. As expected, our baseline Random heuristic performs quite poorly; this result indicates that the node activation set~$S$ has a significant impact on the fixation probability. Finally, recall that Weak Selector is optimal for the regime of weak selection ($\FitAdv\to 0$). 
The fact that Weak Selector underperforms for strong selection ($\FitAdv\to \infty$) indicates an intricate relationship between fitness advantage and fixation probability.

\Paragraph{Weak selection.}
We collect the results on weak selection in \cref{fig:experiments_weak}, using the normalization process above.
Since the derivative satisfies $\fp'(G^S,0)=\sum_{u_i\in S} \alpha(u_i)$ (\cref{subsec:positive_weak}), Lazy Greedy is optimal and coincides with Weak Selector, and is thus omitted from the figure.
Naturally, the Weak Selector always outperforms others, while the Random heuristic is weak. 
The other heuristics have mediocre performance, with a clear advantage of Centrality over the rest, which becomes clearer for larger budget values.

\section{Conclusion}\label{sec:conclusion}

We introduced the positional Moran process and studied the associated fixation maximization problem. 
We have shown that the problem is~$\NPH$ in general, but becomes tractable in the limits of strong and weak selection. 
Our results only scratch the surface of this new process, as several interesting questions are open, such as: Does the strong-selection setting admit a better approximation than the one based on submodularity? 
Can the problem for finite~$\FitAdv$ be approximated within some constant-factor? Are there classes of graphs for which it becomes tractable?

\bibliography{\mypath bibliography}

\begin{thebibliography}{38}
\providecommand{\natexlab}[1]{#1}

\bibitem[{Allen et~al.(2017)Allen, Lippner, Chen, Fotouhi, Momeni, Yau, and
  Nowak}]{Allen2017}
Allen, B.; Lippner, G.; Chen, Y.-T.; Fotouhi, B.; Momeni, N.; Yau, S.-T.; and
  Nowak, M.~A. 2017.
\newblock Evolutionary dynamics on any population structure.
\newblock \emph{Nature}, 544(7649): 227--230.

\bibitem[{Allen et~al.(2021)Allen, Sample, Steinhagen, Shapiro, King, Hedspeth,
  and Goncalves}]{Allen2021}
Allen, B.; Sample, C.; Steinhagen, P.; Shapiro, J.; King, M.; Hedspeth, T.; and
  Goncalves, M. 2021.
\newblock Fixation probabilities in graph-structured populations under weak
  selection.
\newblock \emph{PLOS Computational Biology}, 17(2): 1--25.

\bibitem[{Alman and Williams(2021)}]{Alman2021}
Alman, J.; and Williams, V.~V. 2021.
\newblock A Refined Laser Method and Faster Matrix Multiplication.
\newblock In Marx, D., ed., \emph{Proceedings of the 2021 {ACM-SIAM} Symposium
  on Discrete Algorithms, {SODA} 2021, Virtual Conference, January 10 - 13,
  2021}, 522--539. {SIAM}.

\bibitem[{Ann~Goldberg, Lapinskas, and Richerby(2020)}]{Goldberg2020}
Ann~Goldberg, L.; Lapinskas, J.; and Richerby, D. 2020.
\newblock Phase transitions of the Moran process and algorithmic consequences.
\newblock \emph{Random Structures \& Algorithms}, 56(3): 597--647.

\bibitem[{Antal, Redner, and Sood(2006)}]{Antal2006}
Antal, T.; Redner, S.; and Sood, V. 2006.
\newblock Evolutionary Dynamics on Degree-Heterogeneous Graphs.
\newblock \emph{Phys. Rev. Lett.}, 96: 188104.

\bibitem[{Borgs et~al.(2014)Borgs, Brautbar, Chayes, and
  Lucier}]{Borgs:2014:MSI}
Borgs, C.; Brautbar, M.; Chayes, J.; and Lucier, B. 2014.
\newblock Maximizing Social Influence in Nearly Optimal Time.
\newblock In \emph{SODA}, 946--957.

\bibitem[{Broom et~al.(2010)Broom, Hadjichrysanthou, Rycht{\'a}{\v{r}}, and
  Stadler}]{Broom2010}
Broom, M.; Hadjichrysanthou, C.; Rycht{\'a}{\v{r}}, J.; and Stadler, B. 2010.
\newblock Two results on evolutionary processes on general non-directed graphs.
\newblock \emph{Proceedings of the Royal Society A: Mathematical, Physical and
  Engineering Sciences}, 466(2121): 2795--2798.

\bibitem[{Chatterjee, Ibsen-Jensen, and Nowak(2017)}]{Chatterjee2017}
Chatterjee, K.; Ibsen-Jensen, R.; and Nowak, M.~A. 2017.
\newblock {Faster Monte-Carlo Algorithms for Fixation Probability of the Moran
  Process on Undirected Graphs}.
\newblock In Larsen, K.~G.; Bodlaender, H.~L.; and Raskin, J.-F., eds.,
  \emph{42nd International Symposium on Mathematical Foundations of Computer
  Science (MFCS 2017)}, volume~83 of \emph{Leibniz International Proceedings in
  Informatics (LIPIcs)}, 61:1--61:13. Dagstuhl, Germany: Schloss
  Dagstuhl--Leibniz-Zentrum fuer Informatik.
\newblock ISBN 978-3-95977-046-0.

\bibitem[{Clifford and Sudbury(1973)}]{Clifford1973}
Clifford, P.; and Sudbury, A. 1973.
\newblock A Model for Spatial Conflict.
\newblock \emph{Biometrika}, 60(3): 581--588.

\bibitem[{Demers et~al.(1987)Demers, Greene, Hauser, Irish, Larson, Shenker,
  Sturgis, Swinehart, and Terry}]{Demers1987}
Demers, A.; Greene, D.; Hauser, C.; Irish, W.; Larson, J.; Shenker, S.;
  Sturgis, H.; Swinehart, D.; and Terry, D. 1987.
\newblock Epidemic Algorithms for Replicated Database Maintenance.
\newblock In \emph{Proceedings of the Sixth Annual ACM Symposium on Principles
  of Distributed Computing}, PODC '87, 1–12. New York, NY, USA: Association
  for Computing Machinery.
\newblock ISBN 089791239X.

\bibitem[{D{\'\i}az et~al.(2014)D{\'\i}az, Goldberg, Mertzios, Richerby, Serna,
  and Spirakis}]{Diaz2014}
D{\'\i}az, J.; Goldberg, L.~A.; Mertzios, G.~B.; Richerby, D.; Serna, M.; and
  Spirakis, P.~G. 2014.
\newblock Approximating fixation probabilities in the generalized moran
  process.
\newblock \emph{Algorithmica}, 69(1): 78--91.

\bibitem[{D{\'\i}az et~al.(2016)D{\'\i}az, Goldberg, Richerby, and
  Serna}]{Diaz2016}
D{\'\i}az, J.; Goldberg, L.~A.; Richerby, D.; and Serna, M. 2016.
\newblock Absorption time of the Moran process.
\newblock \emph{Random Structures \& Algorithms}, 49(1): 137--159.

\bibitem[{Even-Dar and Shapira(2007)}]{EvenDar2007}
Even-Dar, E.; and Shapira, A. 2007.
\newblock A Note on Maximizing the Spread of Influence in Social Networks.
\newblock In Deng, X.; and Graham, F.~C., eds., \emph{Internet and Network
  Economics}, 281--286. Berlin, Heidelberg: Springer Berlin Heidelberg.
\newblock ISBN 978-3-540-77105-0.

\bibitem[{Feige(2003)}]{Feige2003}
Feige, U. 2003.
\newblock Vertex cover is hardest to approximate on regular graphs.
\newblock Technical Report MCS03-15, The Weizmann Institute of Science.

\bibitem[{Fountoulakis, Panagiotou, and Sauerwald(2012)}]{Fountoulakis2012}
Fountoulakis, N.; Panagiotou, K.; and Sauerwald, T. 2012.
\newblock Ultra-Fast Rumor Spreading in Social Networks.
\newblock In \emph{Proceedings of the Twenty-Third Annual ACM-SIAM Symposium on
  Discrete Algorithms}, SODA '12, 1642–1660. USA: Society for Industrial and
  Applied Mathematics.

\bibitem[{Galanis et~al.(2017)Galanis, G{\"o}bel, Goldberg, Lapinskas, and
  Richerby}]{Galanis2017}
Galanis, A.; G{\"o}bel, A.; Goldberg, L.~A.; Lapinskas, J.; and Richerby, D.
  2017.
\newblock Amplifiers for the Moran process.
\newblock \emph{Journal of the ACM (JACM)}, 64(1): 5.

\bibitem[{Giakkoupis(2016)}]{Giakkoupis16}
Giakkoupis, G. 2016.
\newblock Amplifiers and Suppressors of Selection for the Moran Process on
  Undirected Graphs.
\newblock \emph{arXiv preprint arXiv:1611.01585}.

\bibitem[{Huang and Traulsen(2010)}]{Huang2010}
Huang, W.; and Traulsen, A. 2010.
\newblock Fixation probabilities of random mutants under frequency dependent
  selection.
\newblock \emph{Journal of Theoretical Biology}, 263(2): 262--268.

\bibitem[{Kaveh, McAvoy, and Nowak(2019)}]{Kaveh2019}
Kaveh, K.; McAvoy, A.; and Nowak, M.~A. 2019.
\newblock Environmental fitness heterogeneity in the Moran process.
\newblock \emph{Royal Society open science}, 6(1): 181661.

\bibitem[{Kempe, Kleinberg, and Tardos(2003)}]{Kempe2003}
Kempe, D.; Kleinberg, J.; and Tardos, E. 2003.
\newblock Maximizing the Spread of Influence through a Social Network.
\newblock In \emph{Proceedings of the Ninth ACM SIGKDD International Conference
  on Knowledge Discovery and Data Mining}, 137--146.

\bibitem[{K{\"o}tzing and Krejca(2019)}]{kotzing2019first}
K{\"o}tzing, T.; and Krejca, M.~S. 2019.
\newblock First-hitting times under drift.
\newblock \emph{Theoretical Computer Science}, 796: 51--69.

\bibitem[{Krause and Golovin(2014)}]{Krause2014}
Krause, A.; and Golovin, D. 2014.
\newblock Submodular function maximization.
\newblock \emph{Tractability}, 3: 71--104.

\bibitem[{Leskovec and Krevl(2014)}]{Leskovec2014}
Leskovec, J.; and Krevl, A. 2014.
\newblock {SNAP Datasets}: {Stanford} Large Network Dataset Collection.
\newblock \//url{http://snap.stanford.edu/data}.

\bibitem[{Li et~al.(2018)Li, Fan, Wang, and Tan}]{IM-Survey}
Li, Y.; Fan, J.; Wang, Y.; and Tan, K. 2018.
\newblock Influence Maximization on Social Graphs: {A} Survey.
\newblock \emph{{IEEE} TKDE}, 30(10): 1852--1872.

\bibitem[{Lieberman, Hauert, and Nowak(2005)}]{Lieberman2005}
Lieberman, E.; Hauert, C.; and Nowak, M.~A. 2005.
\newblock Evolutionary dynamics on graphs.
\newblock \emph{Nature}, 433(7023): 312--316.

\bibitem[{Liggett(1985)}]{Liggett1985}
Liggett, T. 1985.
\newblock \emph{Interacting Particle Systems}.
\newblock Classics in mathematics. Springer New York.
\newblock ISBN 9783540960690.

\bibitem[{McAvoy and Allen(2021)}]{Mcavoy2021}
McAvoy, A.; and Allen, B. 2021.
\newblock Fixation probabilities in evolutionary dynamics under weak selection.
\newblock \emph{Journal of Mathematical Biology}, 82(3): 1--41.

\bibitem[{Minoux(1978)}]{Minoux1978}
Minoux, M. 1978.
\newblock Accelerated greedy algorithms for maximizing submodular set
  functions.
\newblock In Stoer, J., ed., \emph{Optimization Techniques}, 234--243. Berlin,
  Heidelberg: Springer Berlin Heidelberg.
\newblock ISBN 978-3-540-35890-9.

\bibitem[{Monk, Green, and Paulin(2014)}]{Monk2014}
Monk, T.; Green, P.; and Paulin, M. 2014.
\newblock {Martingales and fixation probabilities of evolutionary graphs}.
\newblock \emph{Proc. R. Soc. A Math. Phys. Eng. Sci.}, 470(2165): 20130730.

\bibitem[{Moran(1958)}]{Moran1958}
Moran, P. A.~P. 1958.
\newblock Random processes in genetics.
\newblock \emph{Mathematical Proceedings of the Cambridge Philosophical
  Society}, 54(1): 60–71.

\bibitem[{Nemhauser, Wolsey, and Fisher(1978)}]{Nemhauser1978}
Nemhauser, G.~L.; Wolsey, L.~A.; and Fisher, M.~L. 1978.
\newblock An Analysis of Approximations for Maximizing Submodular Set
  Functions--I.
\newblock \emph{Math. Program.}, 14(1): 265–294.

\bibitem[{Nowak(2006)}]{Nowak2006}
Nowak, M.~A. 2006.
\newblock \emph{{Evolutionary dynamics: exploring the equations of life}}.
\newblock Cambridge, Massachusetts: Belknap Press of Harvard University Press.
\newblock ISBN 0674023382 (alk. paper).

\bibitem[{Pavlogiannis et~al.(2018)Pavlogiannis, Tkadlec, Chatterjee, and
  Nowak}]{Pavlogiannis2018}
Pavlogiannis, A.; Tkadlec, J.; Chatterjee, K.; and Nowak, M.~A. 2018.
\newblock Construction of arbitrarily strong amplifiers of natural selection
  using evolutionary graph theory.
\newblock \emph{Communications Biology}, 1(1): 71.

\bibitem[{Talamali et~al.(2021)Talamali, Saha, Marshall, and
  Reina}]{Talamali2021}
Talamali, M.~S.; Saha, A.; Marshall, J. A.~R.; and Reina, A. 2021.
\newblock When less is more: Robot swarms adapt better to changes with
  constrained communication.
\newblock \emph{Science Robotics}, 6(56): eabf1416.

\bibitem[{Tang, Shi, and Xiao(2015)}]{Tang:2015:IMN}
Tang, Y.; Shi, Y.; and Xiao, X. 2015.
\newblock Influence Maximization in Near-Linear Time: A Martingale Approach.
\newblock In \emph{SIGMOD}, 1539--1554.

\bibitem[{Tang, Xiao, and Shi(2014)}]{Youze:2014}
Tang, Y.; Xiao, X.; and Shi, Y. 2014.
\newblock Influence Maximization: Near-optimal Time Complexity Meets Practical
  Efficiency.
\newblock In \emph{SIGMOD}.

\bibitem[{Tkadlec et~al.(2021)Tkadlec, Pavlogiannis, Chatterjee, and
  Nowak}]{Tkadlec2021}
Tkadlec, J.; Pavlogiannis, A.; Chatterjee, K.; and Nowak, M.~A. 2021.
\newblock Fast and strong amplifiers of natural selection.
\newblock \emph{Nature Communications}, 12(1): 4009.

\bibitem[{Zhang et~al.(2020)Zhang, Zhou, Tao, Karras, Li, and
  Xiong}]{Zhang2020}
Zhang, K.; Zhou, J.; Tao, D.; Karras, P.; Li, Q.; and Xiong, H. 2020.
\newblock \emph{Geodemographic Influence Maximization}, 2764–2774.
\newblock New York, NY, USA: Association for Computing Machinery.
\newblock ISBN 9781450379984.

\end{thebibliography}

\appendix
\section{Appendix}\label{sec:appendix}

\lemmonotonicity*
\begin{proof}
Fix pairs $(S,\FitAdv)$ and $(S',\FitAdv')$.
The key property is that, for any two configurations $\X\subseteq \X'$ and any node $u\in V$, we have
$\Fitness_{\X,S}(u)\le \Fitness_{\X',S'}(u)$: Indeed, the inequality is strict if and only if $u\in (\X'\setminus \X)\cap (S'\setminus S)$.
 As a consequence, Lemma~5 from~\cite{Diaz2016} applies and
yields a desired coupling between the 
processes with parameters $(S,\FitAdv)$ and $(S',\FitAdv')$.
\end{proof}

\lemfixprobstrong*
\begin{proof}
Due to \cref{lem:monotonicity}, it suffices to prove the statement for when $A=\{u\}$ is a singleton set.
Our proof is by a stochastic domination argument.
In particular, we construct a simple Markov chain $\MC$ and a coupling between $\MC$ and the Moran process on $G$ from $\X=\{u\}$ such that the probability of a random walk starting from a particular start state $s_u$ of $\MC$ has at least probability $\fp^{\infty}(G^S, A)$ of getting absorbed in a final state $s_n$.

Let $\ell=2+|\{(u,v)\in E\}|$, i.e., $\ell$ equals 2 plus the number of neighbors of $u$.
$\MC$ consists of the following set of states
 \[
S=\{s_0, s_{u},  s_{u^*}\} \cup \{s_i\colon \ell\leq i \leq n\}\ .
\]
The transition probability function $p\colon S\times S \to[0,1]$ is as follows, for constants $ c_1, c_2$ that depend on $G$ but not on $\FitAdv$.
\begin{align*}
&p(s_u, s_0)=\frac{c_1}{1+\FitAdv}  \qquad  p(s_u, s_{u^*}) = 1-p(s_u, s_0)\\
&p(s_{u^*}, s_{\ell}) = p(s_{i}, s_{i+1}) = c_2 \quad \text{for} \quad \ell \leq i < n\\
&p(s_{u^*}, s_u) = p(s_{i}, s_u)=1-c_2\quad \text{for} \quad \ell \leq i \leq n\\
\end{align*}
while for every state $s$, the remaining probability mass is added as a self-loop probability $p(s,s)$.

We now sketch the coupling between $\MC$ and the positional Moran process on $G$.
The coupling guarantees the following correspondence between the states of $\MC$ and the configuration $\X$ of the Moran process.
\begin{enumerate}
\item For the state $s_u$, we have $u\in \X$.
\item For the state $s_{u^*}$, we have (i)~$u\in X$, and (ii)~for every $(u,v)\in E$, we have $v\in X$.
\item For each state $s_i$, with $\ell \leq i \leq n$, we have (i)~$u\in X$, (ii)~for every $(u,v)\in E$, we have $v\in X$, and (iii)~$|\X|\geq i$.
\end{enumerate}
We argue that the transition probabilities preserve the above correspondence.
\begin{enumerate}
\item While $u\in \X$, the probability for a resident neighbor $v$ of $u$ to reproduce and place an offspring on $u$ is bounded by $\propto \frac{1}{1+\FitAdv}$.
On the other hand, the probability that $u$ reproduces in successive rounds until it places a mutant offspring to each of its neighbors can be exponentially small in $n$ but independent of $\FitAdv$, as, since $u\in S$, the fitness of $u$ is at least a fraction $1/n$ of the total population fitness regardless of the remaining nodes.
It follows that, from a configuration $\X$ with $u\in X$, the probability that $u$ turns resident before every neighbor of $v$ turns mutant is upper-bounded by $\frac{c}{1+\FitAdv}$, for some constant $c$ that depends on $G$ but not on $\FitAdv$.
\item Given any configuration $\X$ with $u\in \X$, the probability that a resident reproduces is at most $\propto \frac{1}{1+\FitAdv}$.
Similarly, the probability that a mutant reproduces and replaces a resident is at least $\propto \frac{1}{1+\FitAdv}$.
Thus, the probability that the latter event happens before the former event is lower-bounded by some constant $c_2$ that depends on $G$ but not on $\FitAdv$.
\end{enumerate}
Thus, we have a coupling in which the probability that a random walk in $\MC$ starting at $s_u$ gets absorbed in $s_n$ is at least as large as $\fp^{\infty}(G^S, \{u\})$.
It is straightforward to verify that the probability of the former event tends to $1$ as $\FitAdv\to \infty$.
Thus $\fp^{\infty}(G^S, \{u\})$, as desired.
\end{proof}

\thmnphard*
\begin{proof}
Following \cref{lem:vertex_cover}, the hardness for $\NodeActivationMoranStrong$ follows directly by reducing the question of whether a graph $G$ has a vertex cover of of size $k$ to the problem of whether $\NodeActivationMoranStrong$ can achieve fixation probability of at least $\frac{n+k}{2n}$.
For the hardness of $\NodeActivationMoran$, recall that for any $S\subseteq V$, the function $\fp(G^S,\FitAdv)$ is continuous on $\FitAdv$.
A crude but slightly more detailed analysis in the proof of \cref{lem:vertex_cover} shows that if we have an edge $(u,v)\in E$ with $u,v\not \in S$, then the fixation probability from $u$ is
\[
\fp^{\infty}(G^{S}, \{u\})< \frac{1}{2} - \frac{1}{2n^4} \ .
\]
In turn, if $S$ is not a vertex cover of $G$ then
\[
\fp^{\infty}(G^{S})<c=\frac{n+|S|}{2n}-\frac{1}{2n^5}
\] 
Due to the continuity of $\fp(G^S,\FitAdv)$, there exists a large enough $\FitAdv^*$ such that if $G$ has a vertex cover $S$ of size $k$ then $\fp(G^S,\FitAdv^*)\geq c$.
On the other hand, due to the monotonicity (\cref{lem:monotonicity}), 
if $G$ has no vertex cover of size $k$, then we have
\[
\max_{S\subseteq V, |S|=k}\fp(G^S,\FitAdv^*) \leq \max_{S\subseteq V, |S|=k} \fp^{\infty}(G^S) < c\ .
\]
Thus for such a fitness advantage $\FitAdv^*$, we can achieve fixation probability $c$ iff $G$ has a vertex cover of size $k$.
\end{proof}

\thmapproxfixprob*
\begin{proof}
We follow the proof strategy of Theorem~11 from~\cite{Diaz2014}.
Given a set $\X$ of nodes occupied by mutants, consider a potential function defined by
$\phi(\X)=\sum_{u\in\X} \frac1{\deg(u)}$, where $\deg(u)$ is the number of edges incident to $u$.
Note that $0\le \phi(\X)\le n$, since $0\le |\X|\le n$ and for each $u\in V$ we have $\deg(u)\ge 1$.
Moreover, given a configuration $\X_t$ at time-point $t$, let $\Delta_t=\phi(\X_{t+1})-\phi(\X_t)$ be a random variable that measures the increase of the potential function $\phi$ in a single step.
We make two claims:
\begin{enumerate}

\item\label{itm:fpras1}
$\E[\Delta_t\mid X_t]\ge 0$: Let $S$ be a set of edges $(u,v)$ such that $u$ is occupied by a mutant and $v$ by a resident, and let $F=\sum_{u\in V} \Fitness_{\X_t,S}(u)$ be the total fitness of the population.
In a single step, the set of nodes occupied by mutants changes either by a mutant replacing a resident neighbor, or by a resident replacing a mutant neighbor. Thus
\begin{align*}
\E[\Delta_t\mid X_t] =&\sum_{(u,v)\in S} \frac{\Fitness_{\X_t,S}(u)}{F}\cdot\frac1{\deg(u)}\cdot\frac{+1}{\deg(v)}\\
&+\sum_{(u,v)\in S}\frac{\Fitness_{\X_t,S}(v)}{F}\cdot\frac1{\deg(v)}\cdot\frac{-1}{\deg(u)} \\
 =& \sum_{(u,v)\in S} \frac{\Fitness_{\X_t,S}(u)-\Fitness_{\X_t,S}(v)}{F\deg(u)\deg(v)}\ge 0,
 \end{align*}
 where the last inequality holds since the fitness of any mutant is at least as large as that of any neighboring resident.

\item\label{itm:fpras2}
If $\emptyset\subsetneq \X_t \subsetneq V$ then $\Pr[\Delta_t\le -1/n \mid \X_t]\ge \frac1{(1+\FitAdv)n^2}$:
If the population is not homogeneous then there exists at least one edge $(u,v)\in E$ with $u\not\in\X_t$ and $v\in\X_t$.
With probability $p_u=\Fitness_{\X_t,S}(u)/F\ge\frac1{(1+\FitAdv)n}$ the agent at $u$ is selected for reproduction and with probability $p_{u\to v}=1/\deg(u)\ge 1/n$ its offspring migrates to $v$, which changes the potential function by $-1/\deg(v)\le -1/n$.
\end{enumerate}

By~\cref{itm:fpras1}, the potential function $\phi$ gives rise to a sub-martingale.
Moreover, the function $\phi$ is bounded by $B=n$ and
by~\cref{itm:fpras2}, we have $V=\operatorname{Var}[\phi(\X_{t+1})\mid \X_t]\ge p_u\cdot p_{u\to v}\cdot (1/n)^2 \ge  \frac1{(1+\FitAdv)n^4}$.
Thus, the standard martingale machinery of drift analysis applies~\cite{kotzing2019first}.
Namely, the rescaled function $\phi\cdot (\phi-2B)+B^2$ satisfies the conditions
of the upper additive drift theorem (with initial value at most $B^2$ and step-wise drift at least $V$).
The expected time until termination is thus at most
\[\ft(G^S,\FitAdv) \le \frac{B^2}{V} \le (1+\FitAdv)n^6. \qedhere
\]
\end{proof}

\thmweakselection*
\begin{proof}

Given a set $S$, we  write $\lambda_i=1$ to indicate that $u_i\in S$ (and $\lambda_i=0$ otherwise).
Given any configuration $\X\subset V$,
we write $x_i=1$ to indicate that $u_i\in\X$ (and $x_i=0$ otherwise).
We also define a function $\phi(\X)= \sum_{i\in[n]} \pi_i \cdot x_i$
and we let
\begin{align}
\Delta(\X,\FitAdv)=\E[\phi(\X_{t+1})-\phi(\X_t)\mid \X_t=\X]
\end{align}
be its expected change in a single step of the Moran process with fitness advantage $\FitAdv$ at nodes $S$.
Let $\overline{x}=\frac{1}{n}\sum_{i\in[n]}\lambda_ix_i$.
Then we can write
\begin{align}
\Delta(\X,\FitAdv) = \sum_{i,j\in[n]} \frac{1+\FitAdv \cdot \lambda_ix_i}{n(1+\FitAdv\cdot \overline{x})} p_{ij}(x_i-x_j)\pi_j.
\end{align}
Finally, the derivative of $\Delta(\X,\FitAdv)$ at $\FitAdv=0$ is
\begin{align*}
\Delta'(\X)  &=\frac{\d}{\d\FitAdv}\Bigr|_{\substack{\FitAdv=0}}  \Delta(\X,\FitAdv)\\
 &=\frac1n \sum_{i,j\in[n]} (\lambda_i \cdot  x_i - \overline{x}) \cdot     p_{ij}\cdot(x_i-x_j)\cdot \pi_j \\
 &=\frac1n \sum_{i,j\in[n]} \lambda_i \cdot x_i \cdot p_{ij}\cdot (x_i-x_j)\cdot \pi_j \\
&\qquad -\frac{\overline{x}}{n} \sum_{i,j\in[n]} p_{ij}\cdot(x_i-x_j)\cdot\pi_j\\
&=\frac1n\sum_{i,j\in[n]} \lambda_i \cdot p_{ij}\cdot \pi_j\cdot x_i(1-x_j),
\numberthis
\end{align*}
where in the last equality we used that $x_i^2=x_i$ and that the second sum vanishes, since all the terms that involve any fixed $x_i$ sum up to $x_i\cdot\sum_{j\in[n]} \left(p_{ij}\cdot\pi_j - p_{ji}\cdot\pi_i\right) =0$ (due to \cref{eq:pi2}).

Now consider the neutral positional Moran process ($\FitAdv=0$) starting with
$\Pr[\X_0=\{u_i\}]=1/n$ for $i\in[n]$.
For any $i,j\in[n]$ and $t\ge 0$, consider an
event $Y_{ij}(t)$ defined as ``at time $t$ we have $u_i\in \X_t$ and $u_j\not\in\X_t$''
and let $\psi_{ij}=\sum_{t=0}^\infty \Pr[Y_{ij}(t)]$ be the total expected time spent with $u_i$ being a mutant and $u_j$ not.
Then by~\cite[Theorem 1]{Mcavoy2021} we have
\begin{align*}
\dfp 
&= \sum_{t=0}^\infty  \sum_{\X\subset V} \Pr[\X_t=\X]\cdot \Delta'(\X) \\
 &= \frac{1}{n}\sum_{i,j\in[n]}\lambda_i \cdot p_{ij} \cdot \pi_j \cdot \psi_{ij}
    = \sum_{u_i\in S}  \alpha (u_i),
\end{align*}

It remains to show that the quantities $\psi_{ij}$ satisfy the linear system of \cref{eq:psi}.
For $i=j$ we clearly have $\psi_{ii}=0$.
For $i\neq j$ we write
\begin{align}\label{eqn:a}
\psi_{ij}  &=\Pr[Y_{ij}(0)] + \sum_{t=0}^{\infty}\Pr[Y_{ij}(t+1)].
\end{align}
The first term is simply $1/n$.
Let $e_{ij}=\frac{1}{n}p_{ij}$ be the probability that, in a single step of the standard Moran process ($\FitAdv=0$), node $u_i$ is selected for reproduction and places its offspring on $u_j$. 
Then we can rewrite each term of the sum in terms of events at time $t$ as follows:
\begin{align*}
\Pr[Y_{ij}(t+1)] &= \sum_{l\in[n]} e_{li}\cdot \Pr[Y_{lj}(t)]+  \sum_{l\in[n]} e_{lj}\cdot \Pr[Y_{il}(t)]\\
&+\bigg(1-\sum_{l\in[n]} \big(e_{li} + e_{lj}\big)\bigg)\cdot \Pr[Y_{ij}(t)].
\end{align*}
Summing over $t$ and plugging this into~\cref{eqn:a} we get
\begin{align*}
 \psi_{ij}= \frac{ \frac1n + \sum_{l\in[n]} e_{li}\cdot \psi_{lj} +  \sum_{l\in[n]} e_{lj}\cdot \psi_{il} }
 {  \sum_{l\in[n]} e_{li}+ \sum_{l\in[n]} e_{lj}  },
\end{align*}
thus we arrive at the the linear system of \cref{eq:psi}.

This concludes the proof.
\end{proof}

\paragraph{Graph from weak-selection experiments.}
\cref{fig:weak_graph} illustrates a challenging graph for $\NodeActivationMoranWeak$, along with the activation choices of different heuristics.

\begin{figure}[!ht]
\includegraphics[scale=0.9]{\mypath 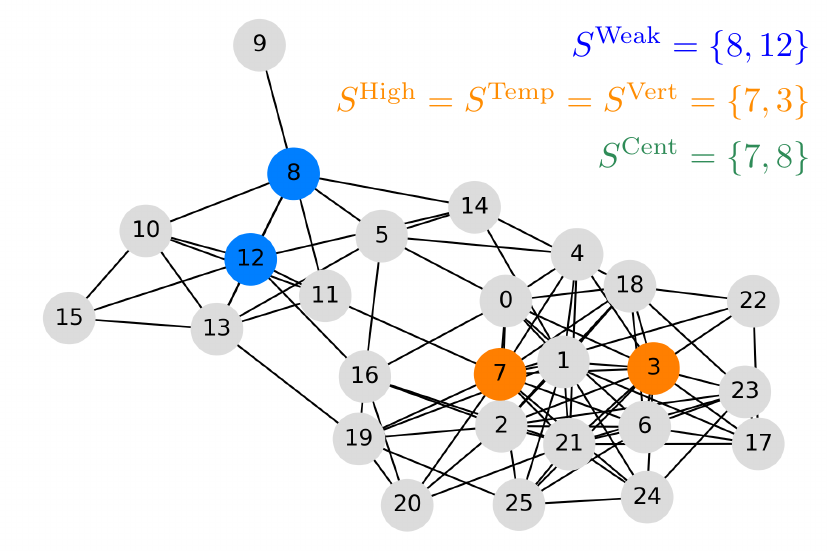}
\caption{
Node activation for weak selection on a challenging graph of $n=26$ nodes,
that shows as an outlier in \cref{fig:experiments_weak}.
The budget is $k=\lfloor 10\% \cdot n \rfloor=2$.
The choices of the different heuristics are shown.
}
\label{fig:weak_graph}
\end{figure}

\end{document}